%% file: NIGApproximationMolCom.tex
\documentclass[journal,comsoc]{IEEEtran}
\pdfoutput=1

\usepackage[T1]{fontenc}
\usepackage[latin9]{inputenc}
\usepackage{epsfig,psfrag}
\usepackage{graphicx}
\usepackage{amsmath,amsfonts,amssymb,amsxtra,bm,amsthm}
\usepackage{color}
\usepackage{paralist, tabularx}
\usepackage{subfigure}
\usepackage{balance}
\usepackage[adjust]{cite}
\usepackage{enumitem}
\usepackage{xfrac}
\usepackage[disable]{todonotes}  
\usepackage{multirow}
\usepackage[final]{changes}
\usepackage{pifont}  
\usepackage{subfigure}
\usepackage{picture,slashbox}

\usepackage{pgfplots, filecontents}
\usetikzlibrary{shapes}
\usepgfplotslibrary{patchplots}

\newcommand{\rmv}{\hspace*{-.3mm}}

\newcommand{\rrmv}{\hspace*{-1mm}}

\newtheorem{Theorem}{Theorem}
\newcommand{\TextRevision}[1]{\textcolor{black}{#1}}
\renewcommand{\baselinestretch}{0.975}

\definecolor{WernerRed}{RGB}{233,71,72}
\definecolor{WernerBlue}{RGB}{87,136,175}
\definecolor{WernerGreen}{RGB}{112,190,109}

\graphicspath{{./Fig/}{../../Matlab/Fig/}}
\DeclareGraphicsExtensions{.pdf}
\begin{document}

\title{\LARGE{Normal Inverse Gaussian Approximation for Arrival Time Difference in Flow-Induced Molecular Communications}}

\author{Werner~Haselmayr,~\IEEEmembership{Member,~IEEE,}
        Dmitry~Efrosinin,~\IEEEmembership{Member,~IEEE,}
        Weisi~Guo,~\IEEEmembership{Member,~IEEE}

}

\maketitle

\input{Abstract}

\section{Introduction}
\label{sec:intro}
\input{Introduction}

\section{Inverse Gaussian Distribution}
\label{sec:sys_mod}
\input{SystemModel}

\section{Normal Inverse Gaussian Approximation}
\label{sec:nig_approx}
\input{NIGApproximation}

\section{Asymptotic Tail Probability}
\label{sec:tail_approx}
\input{TailApproximation}

\section{Applications in Molecular Communications}
\label{sec:appl}
\input{Applications}

\section{Numerical Evaluation}
\label{sec:num_eval}
\input{NumericalEvaluation}

\section{Conclusions}
\label{sec:concl}
\input{Conclusions}

\ifdefined\ACK
  \section{Acknowledgment}
  \input{Acknowledgment}
\fi


\bibliographystyle{IEEEtran}
\bibliography{IEEEabrv,References}

\end{document}

%% file: Abstract.tex

\begin{abstract}
  \TextRevision{In this paper, we consider molecular communications in one-dimensional flow-induced diffusion channels with a perfectly absorbing receiver. In such channels, the random propagation delay until the molecules are absorbed follows an inverse Gaussian~(IG) 
  distribution and is referred to as first hitting time. Knowing the distribution for the difference of the first hitting times of two molecules is very important if the information is encoded by a limited set of molecules and the receiver exploits their arrival time and/or order. }
  Hence, we propose a moment matching approximation by a normal inverse Gaussian~(NIG) distribution and we derive an expression for the asymptotic tail probability. Numerical evaluations showed that the NIG approximation matches very well with the exact solution obtained by numerical convolution of the IG density functions. Moreover, the asymptotic tail probability outperforms state-of-the-art tail approximations.
\end{abstract}

\begin{IEEEkeywords}
  First hitting time, molecular communications, flow-induced diffusion channel, normal inverse Gaussian distribution, inverse Gaussian distribution
\end{IEEEkeywords}


%% file: Introduction.tex

\IEEEPARstart{M}{\lowercase{olecular communications}} (MC) broadly defines information transmission using chemical signals \cite{Farsad16}. Due to its ultra-high efficiency \cite{Rose15ICC} and bio-compatibility it is a
promising candidate for communications at nano-scale. The envisaged applications of MC are in the area of biomedical, environmental and industrial engineering \cite{Nakano2013,Atakan_16}.

In MC the information can be encoded using molecule's concentration \cite{Kuran_12}, number \cite{Lin15}, release time~\cite{Srinivas12}, type \cite{Kim_12} or a combination of the aforementioned methods. The information molecules can be transported from the transmitter to the receiver through pure diffusion, diffusion with flow, active transport (e.g., molecular motors \cite{Farsad_12}) and bacteria \cite{Qiu_17}. 
Currently, most research interest is devoted to diffusion-based MC. In diffusive MC the receivers can be classified as either passive or active~\cite{Noel_16}. A passive receiver only observes the information molecules, whereas an active receiver detects the molecules due to a reaction between the receiver and the molecules. 

The first hitting time denotes the random propagation delay until a molecule arrives at the receiver. 
In case of a perfectly absorbing receiver\footnote{A perfectly absorbing receiver detects and removes the molecule as soon as it arrives at the receiver.} and a one-dimensional\footnote{In a three-dimensional environment, the first hitting time follows a scaled L\'evy distribution for pure diffusion channels~\cite{Yilmaz14}, but is not available for flow-induced diffusion channels.} diffusion channel the first hitting time is distributed according the L\'evy~\cite{Farsad15} and inverse Gaussian (IG)~\cite{Srinivas12} distribution, respectively.


\TextRevision{It is envisioned that MC systems employ nano-machines with very limited capabilities and, thus, it is very likely that molecular signals are represented by a limited set of molecules rather than on the emission of a large number of molecules. In this case, the receiver exploits the arrival time and/or order of the individual molecules, instead of using the received concentration level~\cite{Lin15}. Hence, knowing 
the distribution for the difference of the first hitting times of two molecules becomes important. For example, to determine the out of order arrival probability in type-based information encoding~\mbox{\cite{Shih13,Lin15,Hsieh_13, Haselmayr_18}} or to characterize the noise when the information is encoded in the time duration between two consecutive molecule releases~\cite{Farsad15}. Unfortunately, this distribution has only been derived for pure diffusion channels, following a stable distribution~\cite{Farsad15}.} 

\TextRevision{Hence, in this letter we investigate the distribution for the difference of the first hitting times of two molecules in flow-induced diffusion channels with an absorbing receiver. We propose a moment matching approximation by a normal inverse Gaussian~(NIG) distribution and derive an expression for the asymptotic tail probability. We show numerically that the NIG approximation matches very well with the exact solution derived by numerical convolution of the IG density functions. Moreover, compared to state-of-the-art results~\cite{Shih13}, the presented asymptotic tail probability is more generally applicable and converges faster.}

%% file: SystemModel.tex


In this section, we briefly discuss the main properties of the IG distribution. The probability density function (PDF) of an IG-distributed random variable $X$ is given by \cite{Barndorff_97}
\begin{align}
  f_X(x) & = \frac{a}{\sqrt{2\pi}}\exp(ab)x^{-3/2}\exp\left(-\frac{1}{2}\left(a^2 x^{-1} + b^2 x\right)\right), \ x > 0,
  \label{eq:ig_pdf}
\end{align}

with the parameters $a > 0$ and $b > 0$. We indicate an IG-distributed random variable with the parameters $(a,b)$ by \mbox{$X \sim \text{IG}(a,b)$}. The cumulative distribution function (CDF) can 
be expressed as
\begin{align}
  F_X (x) = & \phi\left(bx^{1/2} - ax^{-1/2}\right) + \exp\left(2ab\right) \nonumber \\
            & + \phi\left(-bx^{1/2} - ax^{-1/2}\right) , \quad x > 0,
  \label{eq:ig_cdf}
\end{align}

with the CDF of the standard normal distribution \mbox{$\phi(x) = 1/\sqrt{2\pi} \int_{-\infty}^x \exp(-t^2/2)\text{d}t$}. Moreover, the tail probability is defined as $\bar{F}_X(x) =  1 - F_X(x)$. The moment-generating function of $X$ is given by
\begin{align}
  M_X(t) 
         & = \exp\left(ab - a\sqrt{b^2 - 2t}\right).
  \label{eq:mom_fct_ig}
\end{align}

%

%% file: NIGApproximation.tex


Let's consider a random variable $Z = X_1 - X_2$, with \mbox{$X_1 \sim \text{IG}(a_1,b_1)$} and $X_2 \sim \text{IG}(a_2,b_2)$. Assuming $X_1$ and $X_2$ are independent, the PDF of $Z$ can be expressed as 
\begin{align}
  f_Z(z)  = \left(f_{X_1}\ast f^{-}_{X_2}\right)(z) 
          = \int\limits_{-\infty}^{\infty} f_{X_1}(u) f_{X_2}(u-z) \text{d}u,
  \label{eq:pdf_ig_diff}
\end{align}

with $f^{-}_{X_2} = f_{X_2}(-x)$. The moment-generating function of $Z$ is given by
\begin{align}
  M_Z(t) & = M_{X_1}(t) M_{X_2}(-t) \nonumber \\
         & = \exp\left(a_1b_1 + a_2b_2 - \left( a_1\sqrt{b_1^2 - 2t} + a_2\sqrt{b_2^2 + 2t} \right) \right).
  \label{eq:mom_fct_ig_diff}
\end{align}

By comparing \eqref{eq:mom_fct_ig} and \eqref{eq:mom_fct_ig_diff} it can be easily verified that $Z$ does not follow an IG distribution.
Moreover, a closed-form expression for the PDF of $Z$ cannot be found, neither through the convolution of the PDFs of $X_1$ and~$X_2$ given in \eqref{eq:ig_pdf}, nor by Laplace transform of the moment-generating function in \eqref{eq:mom_fct_ig_diff}. Thus, we propose an approximation by the NIG distribution. We chose the NIG distribution since it provides a flexible system of distributions, including heavy-tailed and skewness distributions. Moreover, it was shown  in \cite{Eriksson_09}  that the NIG approximation has a smaller approximation error compared to the well-known Gram-Charlier~\cite{Charlier_05} and Edgeworth~\cite{Edgeworth_07} expansion.

The PDF of a NIG-distributed random variable $Y$ is defined by \cite{Eriksson_09}
\begin{align}
  f_Y(y) = & \frac{\alpha \delta}{\pi}\exp\left(\delta \sqrt{\alpha^2 - \beta^2} - \beta (y-\mu)\right) \nonumber \\
           & \times \frac{K_1\left(\alpha\sqrt{\delta^2 + (y-\mu)^2}\right)}{\sqrt{\delta^2 + (y-\mu)^2}},
  \label{eq:nig_pdf}
\end{align}

with the parameters $\alpha > 0$, $\delta > 0$, $\mu \in \mathbb{R}$ and $-\alpha < \beta < \alpha$ and~$K_1(\cdot)$ denotes the  modified Bessel function of the third kind with index $1$. The parameters  $\alpha$, $\beta$, $\mu$ and $\delta$ determine the tail heaviness, asymmetry, location and scaling of the distribution. 
The relation between mean $\mathcal{M}$, variance $\mathcal{V}$, skewness $\mathcal{S}$, excess kurtosis $\mathcal{K}$ of the random variable $Y$ and the four parameters is given by~\cite{Eriksson_09}
\begin{equation}
\begin{aligned}
  \alpha & = 3\rho^{1/2}(\rho - 1)^{-1}\mathcal{V}^{-1/2}|\mathcal{S}|^{-1}  \\
  \beta & =  3(\rho - 1)^{-1}\mathcal{V}^{-1/2}\mathcal{S}^{-1}  \\
  \mu & =  \mathcal{M} - 3\rho^{-1}\mathcal{V}^{1/2}\mathcal{S}^{-1}  \\
  \delta & =   3\rho^{-1}(\rho - 1)^{1/2}\mathcal{V}^{-1/2}|\mathcal{S}|^{-1},
\end{aligned}
\label{eq:nig_moments}
\end{equation}

with $\rho = 3\mathcal{K}\mathcal{S}^{-2} - 4 > 1$. 

We approximate the PDF of $Z$ by matching the mean, variance, skewness and excess kurtosis of $Z$ with the NIG distribution in \eqref{eq:nig_pdf}, \TextRevision{where the skewness and kurtosis are a measure of the asymmetry and the tailedness of the distribution, respectively.} This approach is known as moment matching  method~\cite{Akhiezer_65}. The moments are then used to derive the parameters $\alpha$, $\beta$, $\mu$ and~$\delta$ according to \eqref{eq:nig_moments}. The mean, variance, skewness and excess kurtosis of $Z$ can be expressed in terms cumulants
\begin{equation}
\begin{aligned}
  \hat{\mathcal{M}} & = \kappa_1, \quad
  &\hat{\mathcal{V}} &  = \kappa_2, \\
  \hat{\mathcal{S}}  & = \frac{\kappa_3}{\kappa_2^{3/2}}, \quad
  & \hat{\mathcal{K}} &  = \frac{\kappa_4}{\kappa_2^{2}},
\end{aligned}
\label{eq:rel_moments_cumulants}
\end{equation}

where $\kappa_n$, $n\rrmv \rmv= \rmv\rrmv 1,\ldots,4$, denotes the $n$th cumulant of $Z$. The cumulants can be derived using the moment-generating function of $Z$ (see~\eqref{eq:mom_fct_ig_diff}) as
\begin{align}
  \kappa_n & = \frac{\partial^n}{\partial t^n}\text{ln}M_Z(t)\Bigg|_{t=0}.
  \label{eq:cumulants}
\end{align}

In the following, we present the analytical expressions for the parameters $\alpha$, $\beta$, $\mu$ and $\delta$ for two important use cases in MC (see Section \ref{sec:num_eval}).

\subsection{Use Case 1: $a_1 = a_2 = a$ and $b_1 = b_2 = b$} 
\label{subsec:case_1}
The moment-generating function of $Z$ can be written as
\begin{align}
  M_Z(t) 
         & = \exp\left(a b + a b - \left( a\sqrt{b^2 - 2t} + a\sqrt{b^2 + 2t} \right) \right).
  \label{eq:mom_fct_ig_diff_spec1}
\end{align}


The moments of $Z$ are obtained by applying \eqref{eq:rel_moments_cumulants}  and \eqref{eq:cumulants} and the parameters of the  NIG distribution can be derived using the relation in \eqref{eq:nig_moments}. The parameters can be expressed as
\begin{align}
  \alpha & = \frac{b^2}{\sqrt{5}}\,;\quad \beta = 0\,; \quad \mu = 0\,;\quad \delta = \frac{2}{\sqrt{5}}\frac{a}{b}.
  \label{eq:nig_moments_spec1}
\end{align}

The resulting NIG distribution is symmetric, since $\beta=0$.

\subsection{Use Case 2: $b_1/a_1 = b_2/a_2 = c$}
\label{subsec:case_2}
The moment-generating function of $Z$ can be written as
\begin{align}
  M_Z(t) 
         \rmv = \rmv \exp\left(\left(a_1^2 \rmv + \rmv a_2^2\right)c \rmv - \rmv \left( a_1\sqrt{a_1^2 c^2 - 2t} \rmv + \rmv a_2\sqrt{a_2^2 c^2 \rmv + \rmv2t} \right) \right).
  \label{eq:mom_fct_ig_diff_spec2}
\end{align}


Similar to use case 1, the moments of $Z$ and the parameters of the NIG distribution can be derived by applying \eqref{eq:nig_moments} -- \eqref{eq:cumulants}. The parameters 
are given by
\begin{equation}
\begin{aligned}
  \alpha & = \frac{\left(a_1^2 - a_2^2\right)^2 \sqrt{\left(a_1^4 + 3 a_1^2 a_2^2 + a_2^4\right)\left(a_1^2 - a_2^2\right)^{-2}} }
             {5 a_1^2 a_2^2 \sqrt{\left(a_1^{-2} + a_2^{-2}\right)c^{-3}}\left|\tau\right|}, \\
  \beta &  = \frac{-\left(a_1^2 - a_2^2\right) c^2}{5}, \\
  \mu  & = \frac{a_1^4 - a_2^4}{\left(a_1^4 + 3 a_1^2 a_2^2 + a_2^4\right) c}, \\
  \delta &  = \frac{\sqrt{5} a_1^2 a_2^2 \sqrt{\left(a_1^{-2} + a_2^{-2}\right)c^{-3}}
  }{\sqrt{ a_1^2 a_2^2\left(a_1^2 - a_2^2\right)^{-2}} \left(a_1^4 + 3 a_1^2 a_2^2 + a_2^4\right)\left|\tau \right|},
\end{aligned}
\label{eq:nig_moments_spec2}
\end{equation}

with 
\begin{align*}
  \tau = \frac{a_1^{-4} - a_2^{-4}}{\left(a_1^{-2} + a_2^{-2}\right)^{3/2} \sqrt{c}}.
\end{align*}

The NIG distribution is asymmetric, since $\beta \neq 0$.


%% file: TailApproximation.tex

\TextRevision{In this section, we derive the asymptotic tail probability of the random variable $Z$ given by $\lim_{z \rightarrow \infty }\bar{F}_Z(z)$.}
The tail probability of the random variable $Z$ can be written as
\begin{align}
  \bar{F}_Z(z) = \Pr(Z > z) 
           & = \int\limits_z^\infty \int\limits_{-\infty}^{\infty} f_{X_1}(u) f_{X_2}(u-t) \text{d}u \text{d}t \nonumber \\
           & = \int\limits_{-\infty}^{\infty} f_{X_2}(w) \int\limits_z^\infty f_{X_1}(w+t) \text{d}t \text{d}w \nonumber \\
          & = \int\limits_{-\infty}^{\infty} f_{X_2}(w) \bar{F}_{X_1}(w+z) \text{d}w.
  \label{eq:tail_prob}
\end{align}

\TextRevision{Based on this result, the following theorem presents an expression of the asymptotic tail probability}
\begin{Theorem}
  The asymptotic tail probability of the random  variable $Z = X_1 - X_2$, with \mbox{$X_1 \sim \text{IG}(a_1,b_1)$} and \mbox{$X_2 \sim \text{IG}(a_2,b_2)$}, is given by
  \begin{align}
    \lim_{z \rightarrow \infty }\bar{F}_Z(z) & = \bar{F}_{X_1}(z) M_{X_2}(-b_1^2/2),   
    \label{eq:tail_asym}
  \end{align}
  where $\bar{F}_{X_1}(z)$ denotes the tail probability of $X_1$ and $M_{X_2}(x)$ corresponds to the moment-generating function of $X_2$ defined in \eqref{eq:mom_fct_ig}.
\end{Theorem}

\begin{proof}
  Using the result of $\bar{F}_Z(z)$ in \eqref{eq:tail_prob} and the property $\lim_{z \rightarrow \infty} \bar{F}_{X_1} (z+w)/\bar{F}_{X_1}(z) = \exp(-b_1^2/2 w)$~\cite{Embrechts_83} the 
  asymptotic tail probability can be derived as follows
  \begin{align}
    \lim_{z \rightarrow \infty }\frac{\bar{F}_Z(z)}{\bar{F}_{X_1}(z)} 
      & = \int\limits_{-\infty}^{\infty} \lim_{z \rightarrow \infty } \frac{\bar{F}_{X_1}(z+w) }{\bar{F}_{X_1}(z)}f_{X_2}(w)  \text{d}w \nonumber \\
      & = \int\limits_{-\infty}^{\infty} \exp\left(-b_1^2/2 w\right)f_{X_2}(w)  \text{d}w \nonumber \\
      & = M_{X_2}(-b_1^2/2). \nonumber 
  \end{align}
\end{proof}

%% file: Applications.tex

The approximations proposed in the previous sections are very important for the analysis of many MC systems. In particular, systems which encode information by a limited set of molecules and exploit arrival time and/or order of the individual molecules at the receiver, 
for example time between release and type-based information encoding

\subsection{\TextRevision{Time between Release Information Encoding}}
\label{subsec:timing_chan}
In this case, the information is encoded in the time duration between two consecutive molecule releases \cite{Farsad15}. The arrival time of a single molecule released at time $S$ is given by \mbox{$Y = S + X$}, where $X$ denotes the first hitting time and follows an IG distribution. Let $S_1$ and $S_2$ be the release time of the first and second molecule with $S_2 > S_1$. If the information is encoded in $Z_s = S_2 - S_1$, then, the channel model can be expressed as~\cite{Farsad15}
\begin{align}
        Z_y & = Z_s +  Z_x,
\end{align}

where $Z_y = Y_1 - Y_2$ and \mbox{$Z_x = X_1 - X_2$} denote the difference of arrival and first hitting time, respectively. Moreover, $Z_x$ is the random noise and, thus, its distribution is of interest in order analyze such a system.
Since no closed-form expression can be found for $Z_x$, the approximations proposed in Sections~\ref{sec:nig_approx}~ and~\ref{sec:tail_approx} can be used.

\subsection{\TextRevision{Type-based Information Encoding}}
\label{subsec:crossover_prob}
In this case, the information is encoded in different molecule types. Let's assume that two molecules of different types are released a time interval $T$ apart and that the first hitting time of the first and second released molecule 
 is given by $X_1$ and $X_2$, respectively. Then, the probability that the released molecules arrive out of order can be expressed as
\begin{align}
  \Pr(X_1 -X_2 > T) = \Pr(Z_x > T), 
  \label{eq:crossover_prob}
\end{align}

which is referred to as crossover probability. In order to analyze the crossover probability in \eqref{eq:crossover_prob} the distribution of~$Z_x$ is of interest. Since no closed-form expression can be found for $Z_x$,  the approximations proposed in Sections \ref{sec:nig_approx} and \ref{sec:tail_approx} can be applied. The crossover probability in~\eqref{eq:crossover_prob} can then be used, for example to calculate the error performance~\cite{Haselmayr_18, Shih13}.

%% file: NumericalEvaluation.tex

\color{black}
In this section, we provide a comprehensive numerical evaluation of the approximations proposed in Sections \ref{sec:nig_approx} and~\ref{sec:tail_approx}. We consider a semi-infinite one-dimensional fluid environment with a constant flow of velocity $v$ from transmitter to receiver, which are placed at a distance~$d$. The transmitter is a point source and the receiver a perfectly absorbing boundary. Moreover, we consider the release of two molecules with diffusion coefficients $D_1$ and $D_2$, respectively. The first hitting time of each molecule follows an IG distribution, given by {$X_1 \sim \text{IG}(a_1,b_1)$} and \mbox{$X_2 \sim \text{IG}(a_2,b_2)$}.
The relation between the physical parameters $d$, $v$, $D_1$ and $D_2$ and the parameters $(a_1, b_1)$ and $(a_2, b_2)$ can be expressed as\footnote{Note that \eqref{eq:param_relation} can be obtained by comparing the parametrization of the IG distribution defined in \eqref{eq:ig_pdf} and \cite{Srinivas12}.} 
\begin{equation}
\begin{aligned}
  a_1 & = \frac{d}{\sqrt{2D_1}}, \quad
  & a_2 & = \frac{d}{\sqrt{2D_2}}, \\
  b_1 & = \frac{v}{\sqrt{2D_1}}, \quad
  & b_2 & = \frac{v}{\sqrt{2D_2}}.
\end{aligned}
  \label{eq:param_relation}
\end{equation}

The parameters $(a_1, a_2)$ and  $(b_1, b_2)$ are increased by increasing the distance $d$ and velocity $v$, respectively. Moreover, the aforementioned parameters are increased if the diffusion coefficients are decreased.

\begin{figure}[t!]
  \begin{center}
    \includegraphics[scale = 0.68]{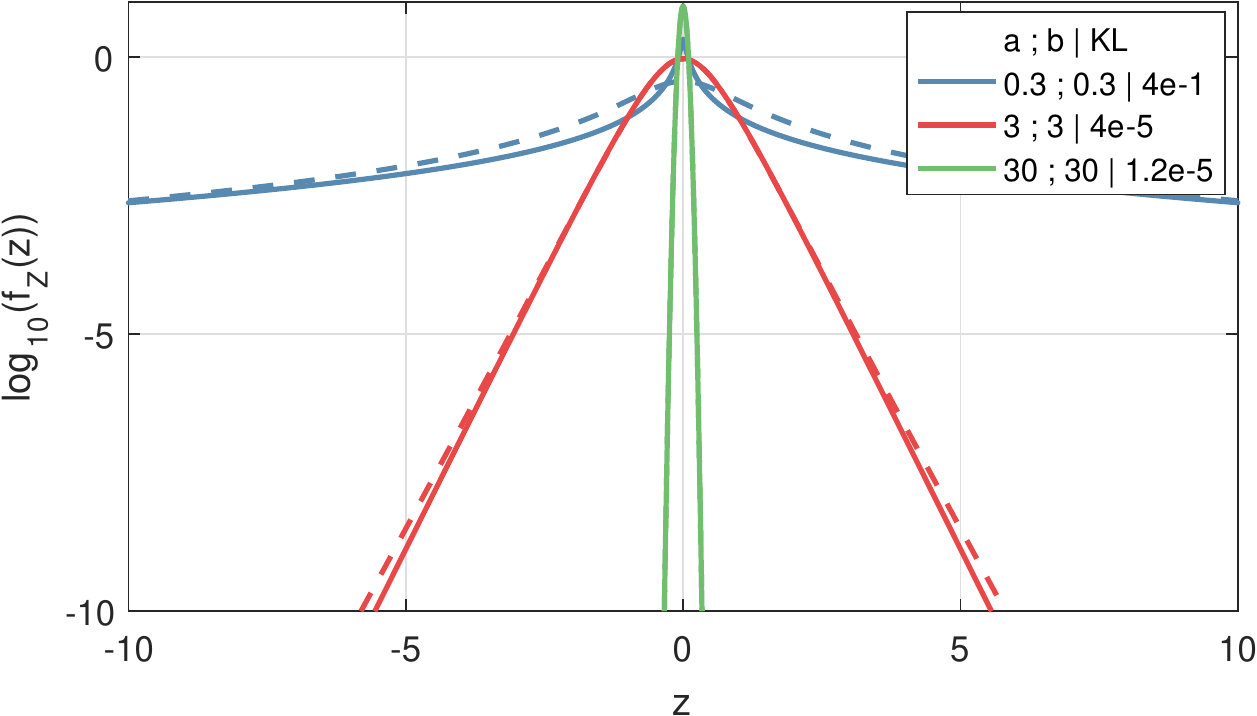}
  \end{center}
  \vspace{-2ex}
    \caption{Probability density function of $Z=X_1 - X_2$ for use case 1 \mbox{(see Section \ref{subsec:case_1})}. Solid lines: Numerical integration of \eqref{eq:pdf_ig_diff}; Dashed lines: NIG approximation with parameters in \eqref{eq:nig_moments_spec1}.}
  \label{fig:pdf_approx_symm_1}
\end{figure}

\begin{figure}[t!]
  \begin{center}
    \includegraphics[scale = 0.68]{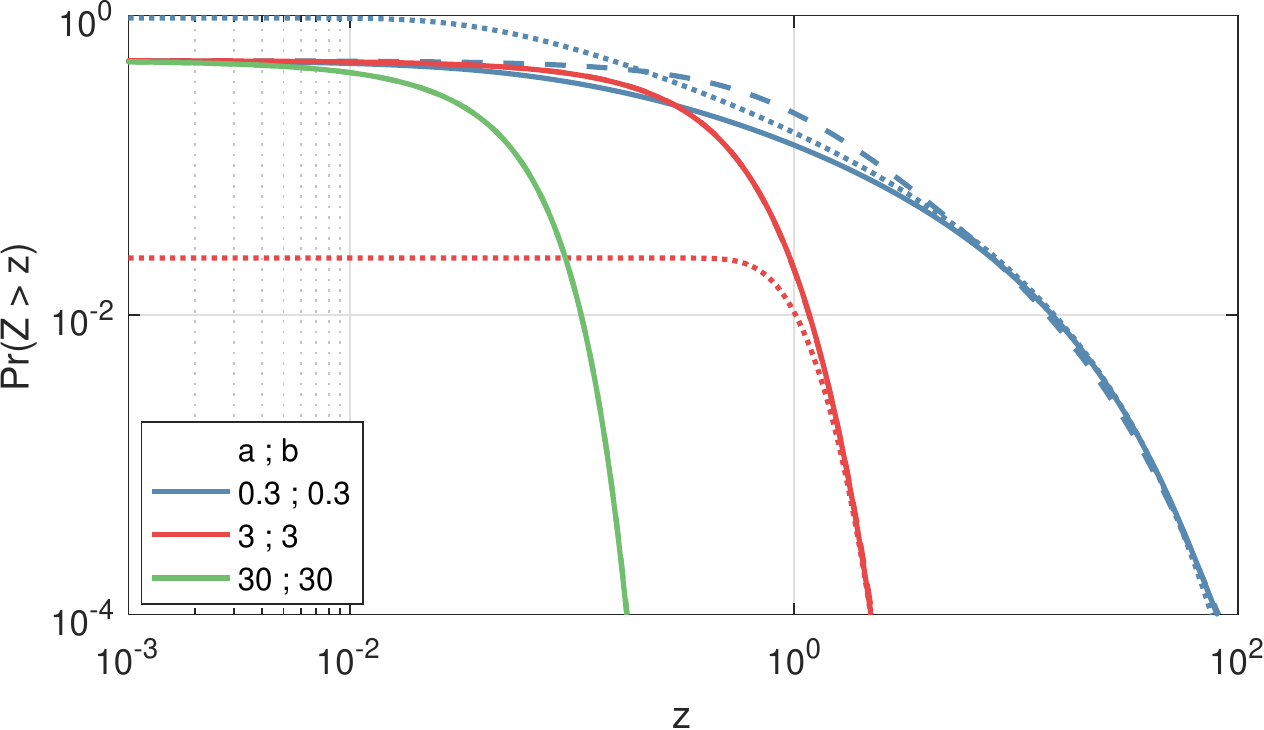}
  \end{center}
  \vspace{-2ex}
    \caption{Tail probability of $Z=X_1 - X_2$ for use case 1 \mbox{(see Section \ref{subsec:case_1})}. Solid lines: Numerical integration of \eqref{eq:tail_prob}; Dashed lines: NIG approximation with parameters in \eqref{eq:nig_moments_spec1}; Dotted lines: Asymptotic tail probability \eqref{eq:tail_asym}.}
  \label{fig:tail_approx_symm_1}
\end{figure}

Figs. \ref{fig:pdf_approx_symm_1} -- \ref{fig:tail_approx_asymm} show the PDF and the tail probability of the random variable $Z =X_1 - X_2$, corresponding to the difference of the first hitting times $X_1$ and $X_2$. In all figures, the solid lines indicate the results obtained through numerical integration of~\eqref{eq:pdf_ig_diff} and the dashed lines correspond to the NIG approximation proposed in Section~\ref{sec:nig_approx}. Additionally, 
in all figures showing the tail probability, the asymptotic results from Section~\ref{sec:tail_approx} are included as dotted lines. The parameters for the numerical evaluation are chosen such that they cover a wide range of  parameters used in related works (e.g., \cite{Kim_14,Shih13,Haselmayr_18}).
In order to measure the difference between the exact probability distribution obtained by numerical integration of~\eqref{eq:pdf_ig_diff} and the NIG approximation we used the Kullback-Leibler~(KL) divergence~\cite{Cover_91}. The results are shown in the legends of Figs.~\ref{fig:pdf_approx_symm_1},~\ref{fig:pdf_approx_symm_2}~and~\ref{fig:pdf_approx_asymm}, where a low value indicates a good match.

Figs.~\ref{fig:pdf_approx_symm_1} -- \ref{fig:tail_approx_symm_2} show the PDF and the tail probability  for use case~1 (see Section~\ref{subsec:case_1}), i.e. $a_1=a_2=a$ and $b_1=b_2=b$ ($D_1=D_2$). In this case the PDF is symmetric. We observe from Fig.~\ref{fig:pdf_approx_symm_1} ($a=b$) that the peak becomes narrow as the parameters increase. Moreover, we observe from Fig.~\ref{fig:pdf_approx_symm_2} that if~$a < b$ results in a narrow peak, while $a > b$ broadens the peak and results in a longer tail. The NIG approximation matches very well especially for large parameters.

From Figs.~\ref{fig:tail_approx_symm_1} and~\ref{fig:tail_approx_symm_2} we observe that the asymptotic tail probability converges to the actual probability if $z$ is sufficiently large. Moreover, we observe a tail probability floor for low values of~$z$, corresponding to $\bar{F}_{Z}(z) = M_{X_2}(-b_1^2/2)$ since $\bar{F}_{X_1}(z) \approx 1$ (see \eqref{eq:tail_asym}). Unfortunately, for larger values of~$a$ and~$b$ this floor is very low (e.g., $1.25\times 10^{-162}$ for $a=b=30$) and, thus, the asymptotic tail probability approaches the actual probability at very low probability values. For these cases the curves for the asymptotic tail probability are not shown for the sake of clarity.

Figs. \ref{fig:pdf_approx_asymm} and \ref{fig:tail_approx_asymm} show the PDF and tail probability for use case~2~\mbox{(see Section \ref{subsec:case_2})}, i.e. $b_1/a_1 = b_2/a_2$ ($D_1 \neq D_2$). In this case the PDF is asymmetric. 
We observe from Fig.~\ref{fig:pdf_approx_asymm} a positive skew (right tail is longer) if $a_1 < a_2$ and a negative skew (left tail is longer) if $a_1 > a_2$ . Moreover, we observe that the NIG approximation matches very well with the exact probability distribution.

Similar to use case 1, we observe from Fig.~\ref{fig:tail_approx_asymm} that the asymptotic tail probability converges to the actual probability, if $z$ is sufficiently large. Again, we do not show the curves for the asymptotic tail probability if the 
tail probability floor is very low.

\begin{figure}[t!]
  \begin{center}
    \includegraphics[scale = 0.68]{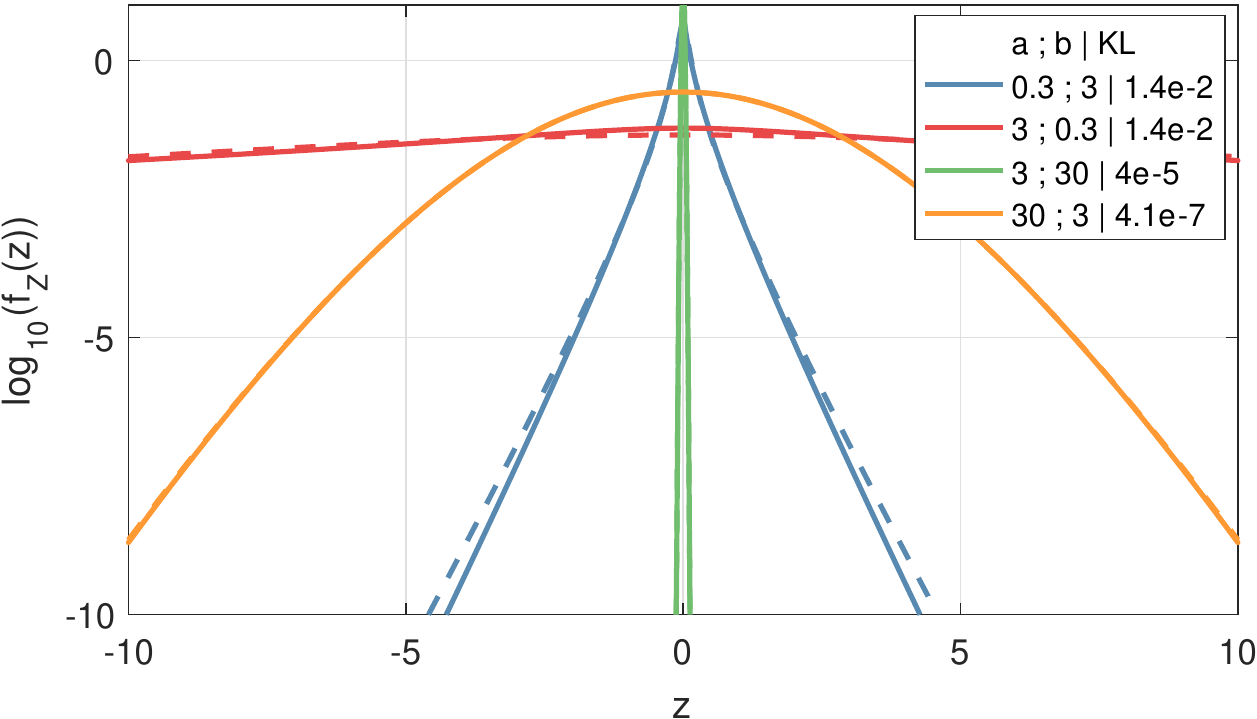}
  \end{center}
  \vspace{-2ex}
    \caption{Probability density function of $Z=X_1 - X_2$ for use case 1 \mbox{(see Section \ref{subsec:case_1})}. Solid lines: Numerical integration of \eqref{eq:pdf_ig_diff}; Dashed lines: NIG approximation with parameters in \eqref{eq:nig_moments_spec1}.}
  \label{fig:pdf_approx_symm_2}
\end{figure}

\begin{figure}[t!]
  \begin{center}
    \includegraphics[scale = 0.68]{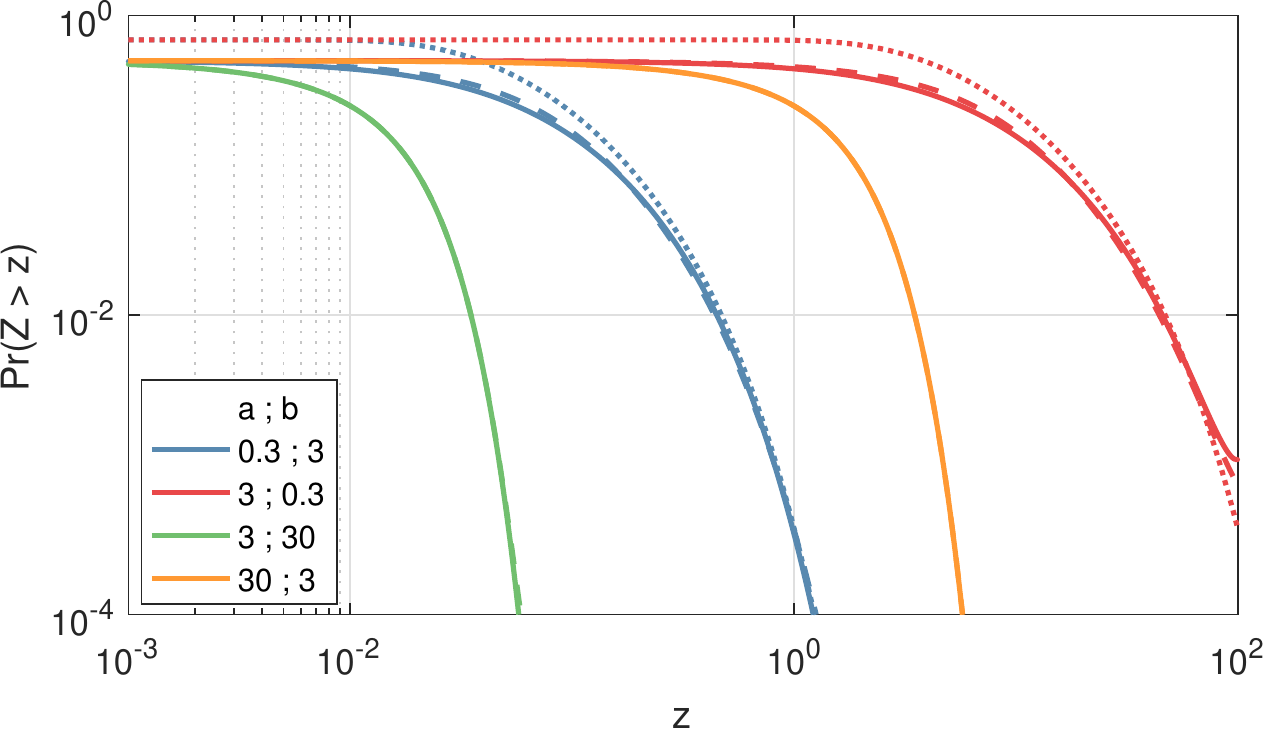}
  \end{center}
  \vspace{-2ex}
    \caption{Tail probability of $Z=X_1 - X_2$ for use case 1 \mbox{(see Section \ref{subsec:case_1})}. Solid lines: Numerical integration of \eqref{eq:tail_prob}; Dashed lines: NIG approximation with parameters in \eqref{eq:nig_moments_spec1}; Dotted lines: Asymptotic tail probability \eqref{eq:tail_asym}.}
  \label{fig:tail_approx_symm_2}
\end{figure}

In Fig. \ref{fig:tail_approx_symm_comp} we compare the asymptotic tail probability in~\eqref{eq:tail_asym} with a recently proposed approximation of the asymptotic tail probability~\cite{Shih13}
\begin{align}
  \bar{F}_Z(z) = \frac{2}{b^2}\exp\left(-\left(\sqrt{2}-1\right)ab\right)f_{X_1}(x),
  \label{eq:eq:tail_asym_soa}
\end{align}

\TextRevision{where $f_{X_1}(x)$ denotes the PDF of an IG distribution defined in~\eqref{eq:ig_pdf}. We observe that the asymptotic tail probability in~\eqref{eq:tail_asym} converges slightly faster than the approximation in~\eqref{eq:eq:tail_asym_soa}. Moreover, in contrast to \eqref{eq:eq:tail_asym_soa},  the approximation in \eqref{eq:tail_asym} is applicable for use case 1 and 2.}

\begin{figure}[t!]
  \begin{center}
    \includegraphics[scale = 0.68]{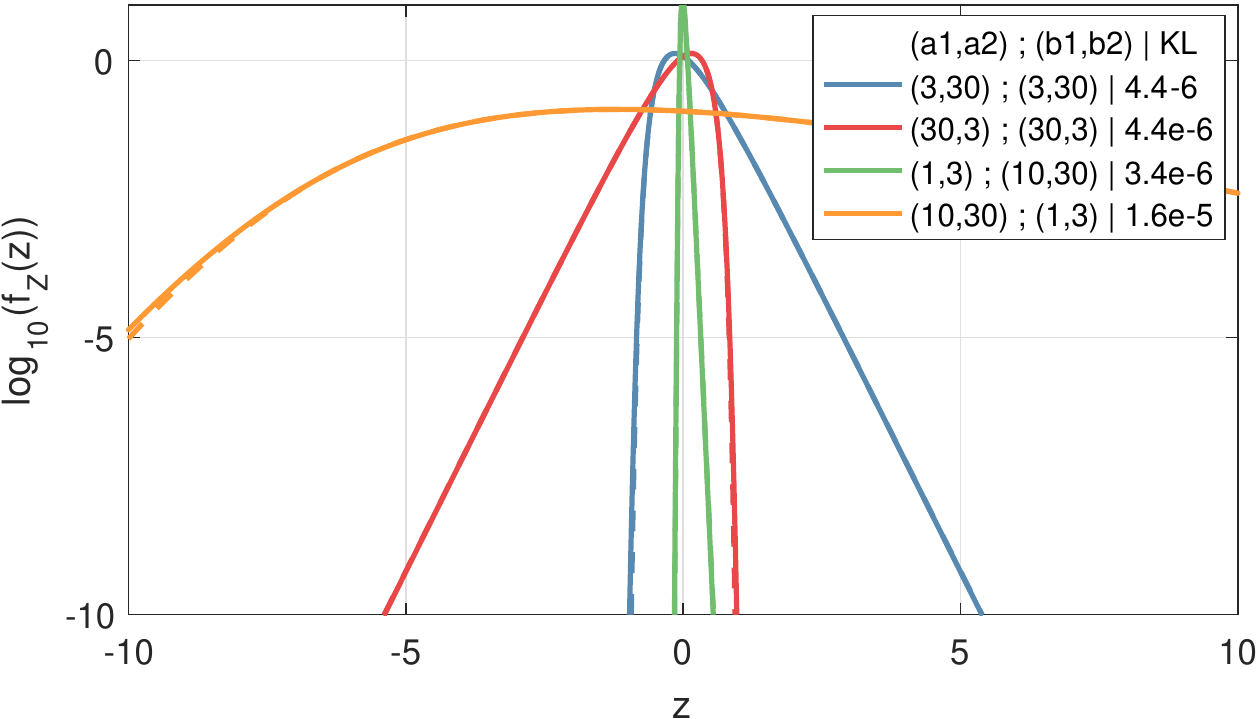}
  \end{center}
  \vspace{-2ex}
    \caption{Probability density function of $Z=X_1 - X_2$ for use case 2 \mbox{(see Section \ref{subsec:case_2})}. Solid lines: Numerical integration of \eqref{eq:pdf_ig_diff}; Dashed lines: NIG approximation with parameters in \eqref{eq:nig_moments_spec2}.}
  \label{fig:pdf_approx_asymm}
\end{figure}

\begin{figure}[t!]
  \begin{center} 
    \includegraphics[scale = 0.68]{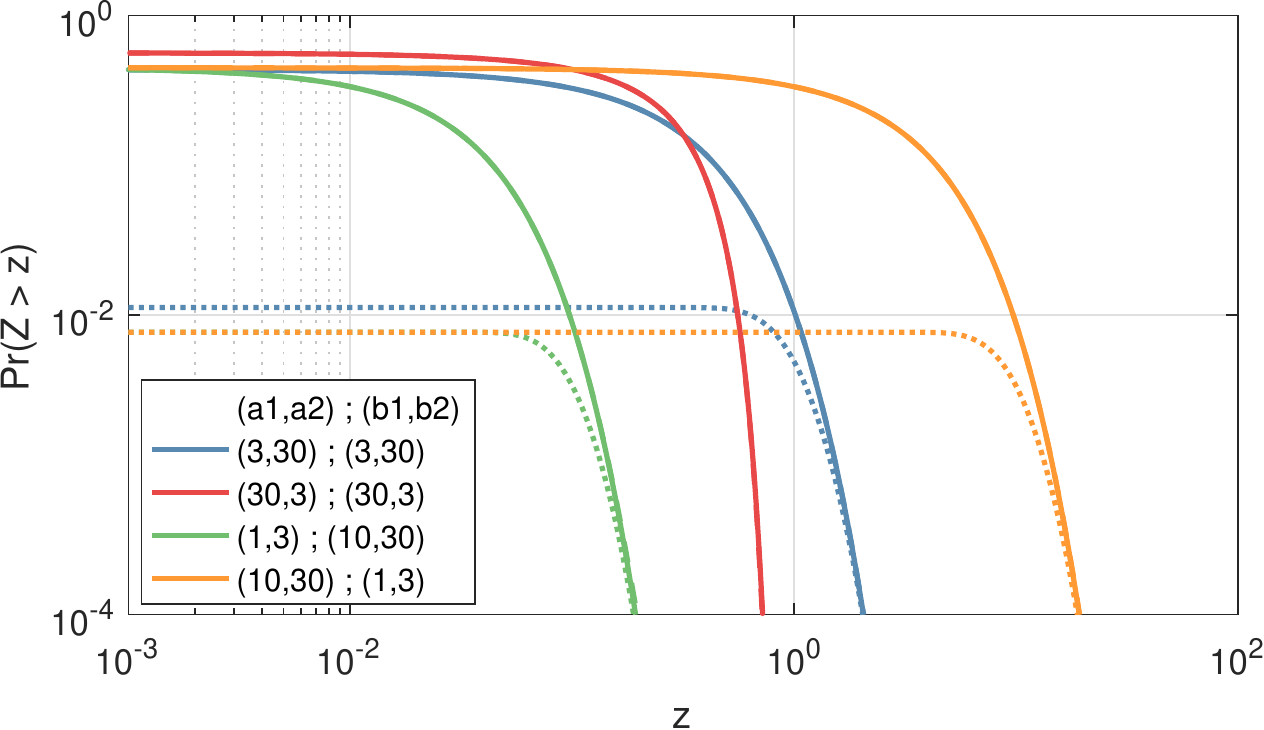}
  \end{center}
  \vspace{-2ex}
    \caption{Tail probability of $Z=X_1 - X_2$ for use case 2 \mbox{(see Section \ref{subsec:case_2})}.  Solid lines: Numerical integration of \eqref{eq:tail_prob}; Dashed lines: NIG approximation with parameters in \eqref{eq:nig_moments_spec2}; Dotted lines: Asymptotic tail probability \eqref{eq:tail_asym}.}
  \label{fig:tail_approx_asymm}
\end{figure}

\begin{figure}[t!]
  \begin{center}
    \includegraphics[scale = 0.68]{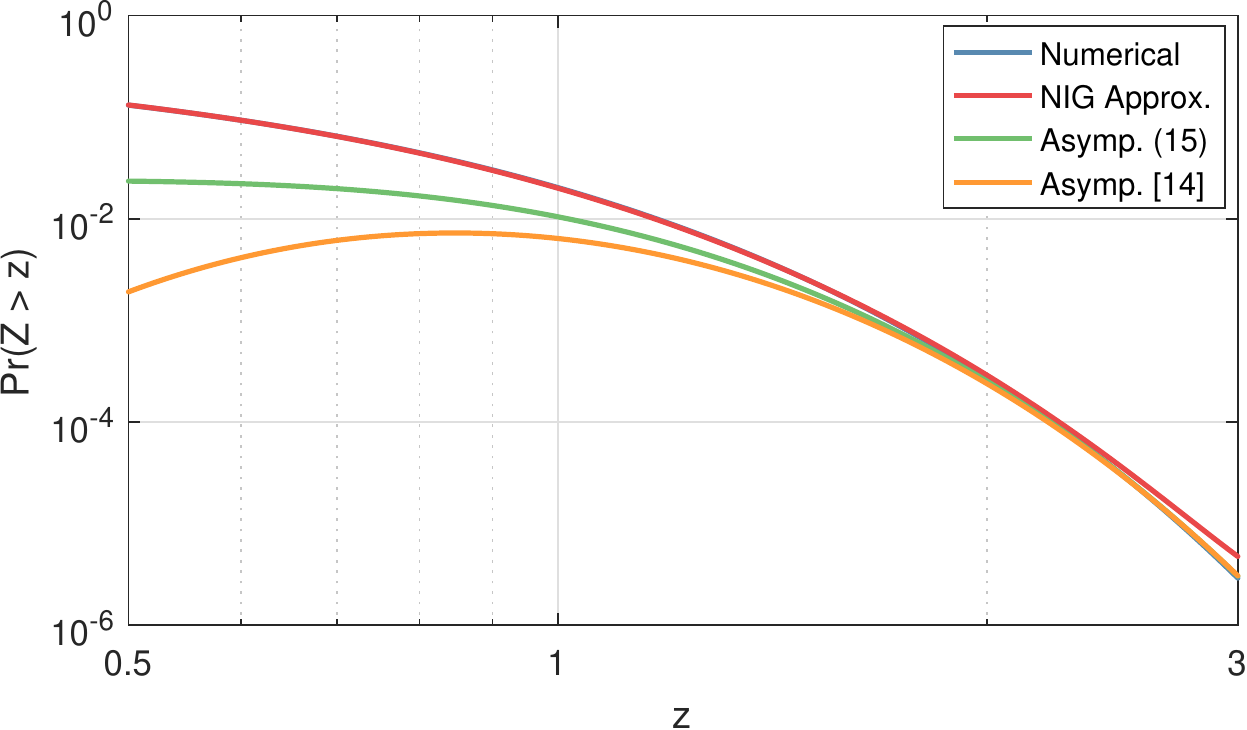}
  \end{center}
  \vspace{-2ex}
    \caption{Tail probability comparison of $Z=X_1 - X_2$ for use case 1 \mbox{(see Section \ref{subsec:case_1})} with $(a_1,a_2) = (b_1,b_2) = (3,3)$.}
  \label{fig:tail_approx_symm_comp}
\end{figure}

\color{black}

%% file: Conclusions.tex

In this letter, we investigated the distribution for the difference of the first hitting times of two molecules in one-dimensional flow-induced diffusion channels. We proposed a moment-matching approximation by a NIG distribution and derived an expression for the asymptotic tail probability. Numerical evaluations confirmed the NIG approximation and showed that the asymptotic tail probability converges faster than state-of-the-art approximations. \TextRevision{We showed that the proposed approximations are very important for the analysis of MC systems which encode information by a limited set of molecules and exploit arrival time and/or order of the individual molecules at the receiver. In a future work, the presented results could serve as a basis for the analysis of more complex MC systems in flow-induced diffusion channels (see spatially distributed~MC~\mbox{\cite{Deng_17 ,Arifler_17,Zabini_18})}.}

\balance
